\documentclass[onecolumn,12pt]{IEEEtran}
\usepackage{graphicx,epic,eepic,epsfig,amsmath,amsfonts,amsthm,latexsym,amssymb,color,enumerate,bbm}
\allowdisplaybreaks[3]

\usepackage[colorlinks,
            linkcolor=blue,
            anchorcolor=blue,
            citecolor=red,
            ]{hyperref}

\newtheorem{theorem}{Theorem}

\newtheorem{corollary}[theorem]{Corollary}
\newtheorem{lemma}[theorem]{Lemma}
\newtheorem{proposition}[theorem]{Proposition}
\newtheorem{example}[theorem]{Example}
\theoremstyle{definition}
\newtheorem{definition}[theorem]{Definition}
\theoremstyle{remark}

\newcommand{\nc}{\newcommand}
\nc{\tr}{\operatorname{Tr}}
\nc{\ket}[1]{|#1\rangle}
\nc{\bra}[1]{\langle#1|}
\nc{\proj}[1]{| #1\rangle\!\langle #1 |}
\nc{\beq}{\begin{equation}}
\nc{\eeq}{\end{equation}}
\nc{\rar}{\rightarrow}
\nc{\lrar}{\longrightarrow}
\nc{\1}{\mathbbm{1}}
\nc{\supp}{\operatorname{supp}}
\nc{\ox}{\otimes}
\nc{\nb}{\nonumber}
\nc{\id}{{\operatorname{id}}}
\nc{\di}{\mathrm{d}}

\nc{\cD}{\mathcal{D}}  \nc{\cE}{\mathcal{E}}  \nc{\cH}{\mathcal{H}}
\nc{\cM}{\mathcal{M}}  \nc{\cN}{\mathcal{N}}  \nc{\cP}{\mathcal{P}}
\nc{\cS}{\mathcal{S}}  \nc{\cT}{\mathcal{T}}
\nc{\bR}{\bar{R}}  \nc{\tR}{\tilde{R}}  \nc{\bB}{\bar{B}}

\begin{document}
\title{Reliable Simulation of Quantum Channels: the Error Exponent}

\author{Ke~Li, Yongsheng~Yao

\thanks{The work of Ke Li was supported by the National Natural Science Foundation
of China (No. 61871156, No. 12031004). The work of Yongsheng Yao was supported
by the National Natural Science Foundation of China (No. 61871156, No. 12071099)
and the European Research Council (ERC Grant Agreement No. 948139).

Ke Li is with the Institute for Advanced Study in Mathematics,
Harbin Institute of Technology, Nangang District, Harbin 150001,
China (e-mail:carl.ke.lee@gmail.com).

Yongsheng Yao is with the Institute for Quantum Information,
RWTH Aachen University, Aachen 52074, Germany, and
he was with the Institute for Advanced Study in Mathematics and School
of Mathematics, Harbin Institute of Technology when this work was
done (e-mail:yongsh.yao@gmail.com).}}
\date{}

\maketitle
\begin{abstract}
The Quantum Reverse Shannon Theorem has been a milestone in quantum information theory. It states that asymptotically reliable simulation of a quantum channel, assisted by unlimited shared entanglement, requires a rate of classical communication equal to the channel's entanglement-assisted classical capacity. In this paper, we study the error exponent of quantum channel simulation, which characterizes the optimal speed of exponential convergence of the performance towards the perfect, as the blocklength increases. Based on channel purified distance, we derive lower and upper bounds for the error exponent. Then we show that the two bounds coincide when the classical communication rate is below a critical value, and hence, we have determined the exact formula of the error exponent in the low-rate case. This enables us to obtain an operational interpretation to the channel's sandwiched R\'enyi information of order from 1 to 2, since our formula is expressed as a transform of this quantity. In the derivation, we have also obtained an achievability bound for quantum channel simulation in the finite-blocklength setting, which is of realistic significance.
\end{abstract}

\begin{IEEEkeywords}
quantum channel simulation,
error exponent,
reliability function,
sandwiched R\'enyi information,
purified distance
\end{IEEEkeywords}

\section{Introduction}
\label{sec:introduction}
Any dynamical process in nature can be modeled as a quantum channel. The simulation of quantum channels is not only a fundamental problem in quantum communication~\cite{BDHSW2014quantum,BCR2011the,
BSST2002entanglement,Devetak2006triangle,BBCW2013entanglement,FWTB2020quantum,TWH2020application,
CLMW2011zero,DuanWinter2015no}, but has also attracted broad interest in the topics of quantum thermodynamics~\cite{FBB2019thermodynamic,FBB2021thermodynamic} as well as general quantum resource theory~\cite{GourScandolo2020dynamical,RegulaTakagi2021one,ChitambarGour2019quantum}. Dual to the ordinary scenario of distilling perfect resources from noisy channels, channel simulation is the problem of generating a noisy channel by making use of noiseless resources.

Reverse Shannon simulation of a quantum channel refers to the task of simulating a noisy quantum channel by consuming noiseless classical communication, assisted by an unlimited amount of free shared entanglement. The celebrated Quantum Reverse Shannon Theorem~\cite{BDHSW2014quantum} states that for asymptotically reliable simulation of a quantum channel $\cN_{A\rar B}$ from system $A$ to system $B$, an amount of classical communication equal to the channel's entanglement-assisted classical capacity $C_E(\cN)$, is necessary and sufficient~\cite{BDHSW2014quantum,BCR2011the}. The entanglement-assisted classical capacity, in turn, is given by the channel's mutual information~\cite{BSST2002entanglement,BSST1999entanglement}
\[
C_E(\cN)=I(\cN):=\max_{\varphi_{RA}}I(R:B)_{\cN(\varphi_{RA})},
\]
where the maximization is over all pure states $\varphi_{RA}$ on system $A$ and a reference system $R$, and $I(R:B)_{\cN(\varphi_{RA})}$ is the mutual information of the bipartite state $N_{A\rar B}(\varphi_{RA})$. Being one of the foundational components of quantum Shannon theory, the Quantum Reverse Shannon Theorem establishes that $C_E$ is the unique quantity to characterize reversible conversion of one channel to another assisted by free entanglement, and it has found applications in proving strong converse theorems~\cite{BDHSW2014quantum}, rate-distortion theorems~\cite{DHW2012quantum,WDHW2013quantum}, as well as results in network information theory~\cite{HsiehShun2016channel}.

While the Quantum Reverse Shannon Theorem proves that $I(\cN)$ is the sharp threshold above which asymptotically reliable simulation of $\cN_{A\rar B}$ is possible, it does not answer the question regarding how fast perfect simulation can be approached as the copy of channels grows. This is the goal pursued by the study of the error exponent, also known as reliability function, which characterizes the best rate of exponential convergence of the performance of an information processing task towards the perfect~\cite{Gallager1968information}. The study of reliability functions in quantum information dates back to more than two decades ago~\cite{BurnashevHolevo1998on,Winter1999coding,Holevo2000reliability}, and progresses have been made ever since~\cite{Hayashi2007error,Dalai2013lower,Hayashi2015precise,
DalaiWinter2017constant,QWW2018applications,MBGYA2019a,CHT2019quantum,CHDH2020non,Dupuis2021privacy}. Recently, exact formulas have been obtained for certain quantum state processing tasks, including quantum privacy amplification~\cite{LYH2023tight}, quantum state merging~\cite{LiYao2024reliability}, and source compression with quantum side information~\cite{Renes2022achievable}. However, precise evaluation of the reliability functions for general quantum channels was unknown prior to the present work.

In this paper, we study the reliability function of reverse Shannon simulation of a quantum channel $\cN_{A\rar B}$. We use purified distance as the measure of performance, which, being a function of the fidelity, is a natural distance measure in quantum information. Building on the preceding work~\cite{LiYao2024reliability} of the authors, we derive tight upper and lower bounds for the reliability function. We show that these two bounds coincide when the rate of classical communication cost is below a critical value $R_{\textrm{critical}}$. Thus, in the low-rate case, we have determined the exact formula of the reliability function. In the derivation, we have obtained an explicit achievability bound for the performance of simulating $\cN_{A\rar B}^{\ox n}$ for a finite blocklength $n$, which is of realistic significance. These results are given in terms of the sandwiched R\'enyi divergence~\cite{MDSFT2013on,WWY2014strong}. More specifically, the channel's sandwiched R\'enyi information~\cite{GuptaWilde2015multiplicativity}, defined as follows, plays a crucial role:
\[
  I_{\alpha}(\cN):=\max_{\varphi_{RA}}I_\alpha(R:B)_{\cN(\varphi_{RA})},
\]
where the maximization is over all pure state $\varphi_{RA}$, and $I_\alpha(R:B)_{\cN(\varphi_{RA})}$ is the sandwiched R\'enyi mutual information of the bipartite state $N_{A\rar B}(\varphi_{RA})$.

\emph{Note added:} an earlier version of the present work was posted on arXiv with a slightly different title (see~\cite{LiYao2021reliable}). In an independent and concurrent work~\cite{RTB2023moderate}, Ramakrishnan, Tomamichel and Berta have derived the moderate deviation expansion for quantum channel simulation. While the results of the work~\cite{RTB2023moderate} and the present one are largely different, there are some overlaps in technical details: the two works both employ channel purified distance to measure the performance, and they both make use of the de Finetti reduction based on purified distance.

Since our original arXiv post, there have been more activities in understanding the error exponents of quantum tasks. Renes~\cite{Renes2022achievable}, combined with an earlier work~\cite{CHDH2020non}, has obtained the reliability function of source compression with quantum side information. Cheng and Gao have proven tight exponential achievability bounds, based on trace distance, for quantum tasks related to soft covering and convex splitting~\cite{ChengGao2024error,ChengGao2023tight}. More recently, Beigi and Tomamichel~\cite{BeigiTomamichel2023lower} have derived new achievability bounds for classical communication over classical-quantum channels, as well as general channels with entanglement assistance.

The remainder of this paper is organized as follows. In Section~\ref{sec:notation-def} we deal with the notation and definitions. In particular, we formulate the definition of the main problem. In Section~\ref{sec:symmetrizaion} we introduce the technique of symmetrization and de Finetti reduction in channel simulation. In Section~\ref{sec:main-results} we present the main results. The proofs are given in the subsequent two sections, where in Section~\ref{sec:Proofs-achi} we deal with the achievability part, and in Section~\ref{sec:Proofs-converse} we treat the converse part. Finally, in Section~\ref{sec:discussion} we conclude the paper with some discussion.

\section{Notation and Definitions}
\label{sec:notation-def}
\subsection{Basic Notation}
Let $\cH$ be a Hilbert space. $\cP(\cH)$ denotes the set of positive semidefinite operators on $\cH$.
$\cS(\cH):=\{\rho \in \cP(\cH) | \tr\rho=1 \}$ and $\cS_{\leq}(\cH):=\{\rho \in \cP(\cH) | \tr\rho\leq 1\}$ are the set of quantum states and subnormalized quantum state, respectively. The Hilbert space associated with a quantum system $A$ is denoted by $\cH_A$. $|A|$ stands for the dimension of $\mathcal{H}_A$. We use $\1_A$ to denote the identity operator on $\cH_A$. For two operators $X,Y$ on $\cH$, we say that $X\geq Y$ if $X-Y\in\cP(\cH)$. The notations $\cP(\cH_A)$, $\cS(\cH_A)$ and $\cS_{\leq}(\cH_A)$ are abbreviated as $\cP(A)$, $\cS(A)$ and $\cS_{\leq}(A)$, respectively. The pure state $\proj{\varphi}\in\cS(\cH)$, usually abbreviated as $\varphi$, is also represented by the unit vector $\ket{\varphi}\in\cH$.

A quantum channel $\cN_{A\rar B}$ is a completely positive and trace-preserving (CPTP) map that sends states on the input Hilbert space $\cH_A$ to states on the output Hilbert space $\cH_B$. It is convenient to use the Stinespring dilation, which tells that there is an environment system $E$ and an isometry $U_\cN:\cH_A\rar\cH_B\ox\cH_E$ such that $\cN(\rho)=\tr_EU_\cN(\rho)$ for any input state $\rho_A$. Here we have used the notation $U_\cN(\rho)\equiv U_\cN\rho U_\cN^*$. An identity channel acting on system $R$ is denoted by $\id_R$.

The purified distance between two quantum states $\rho$ and $\sigma$ is defined as $P(\rho,\sigma)=\sqrt{1-F^2(\rho,\sigma)}$, where $F(\rho,\sigma)=\|\sqrt{\rho}\sqrt{\sigma}\|_1$ is the fidelity. The channel purified distance between two quantum channels $\cM_{A\rar B}$ and $\cN_{A\rar B}$ is given by
\[P\left(\cM_{A\rar B}, \cN_{A\rar B}\right)
:=\max_{\varphi_{RA}}P\left(\id_R\ox\cM(\varphi_{RA}),\id_R\ox\cN(\varphi_{RA})\right),\]
where we assume $|R|=|A|$ w.l.o.g., and the maximization is over all pure states. Note that for
simplicity, sometimes we omit the subscript indicating which system the channel acts on, and we even do not write the identity channel explicitly. For example, $N(\varphi_{RA})$ means $\id_R\ox\cN_{A\rar B}(\varphi_{RA})$.

Throughout this paper, $\exp$ and $\log$ are both with base 2.

\subsection{Quantum Entropies}
The sandwiched R\'enyi information quantities are with order $\alpha\in [\frac{1}{2},1)\cup (1,\infty)$. For quantum states $\rho, \sigma \in \cS(\cH)$, the sandwiched R\'enyi divergence~\cite{MDSFT2013on,WWY2014strong} is defined as
\[
D_{\alpha}(\rho\|\sigma):=\frac{1}{\alpha-1} \log \tr\big(\sigma^{\frac{1-\alpha}{2\alpha}}\rho\sigma^{\frac{1-\alpha}{2\alpha}}\big)^\alpha
\]
if either $\alpha\in (1,\infty)$ and $\supp(\rho)\subseteq\supp(\sigma)$ or $\alpha\in [\frac{1}{2},1)$ and $\supp(\rho)\not\perp\supp(\sigma)$, otherwise we set $D_{\alpha}(\rho\|\sigma)=+\infty$. For a bipartite quantum state $\rho_{AB} \in \cS(AB)$, the sandwiched R\'enyi mutual information is defined as~\cite{WWY2014strong, Beigi2013sandwiched}
\[
I_{\alpha}(A:B)_\rho:=\min_{\sigma_B\in\cS(B)} D_{\alpha}(\rho_{AB} \| \rho_A \ox \sigma_B).
\]
For a quantum channel $\cN_{A \rar B}$, its sandwiched R\'enyi information~\cite{GuptaWilde2015multiplicativity} is defined as
\[
  I_{\alpha}(\cN):=\max_{\varphi_{RA}}I_\alpha(R:B)_{\cN_{A\rar B}(\varphi_{RA})},
\]
where the maximization is over all pure states $\varphi_{RA}$.

When $\alpha\rar 1$, the above R\'enyi quantities converges respectively, to the quantum relative entropy, mutual information, and channel mutual information:
\begin{align*}
D(\rho\|\sigma) &:=\tr(\rho(\log\rho-\log\sigma)), \\
I(A:B)_\rho &:=D(\rho_{AB} \| \rho_A \ox \rho_B),  \\
I(\cN_{A \rar B})&:=\max_{\varphi_{RA}}I(R:B)_{\cN_{A\rar B}(\varphi_{RA})}.
\end{align*}

\subsection{Reverse Shannon Simulation of Quantum Channels}
The quantum reverse Shannon simulation concerns how to use classical communication and unlimited entanglement to simulate any quantum channel $\mathcal{N}_{A \rightarrow B}$.
\begin{definition}[\em{reverse Shannon simulation}]
Let $\mathcal{N}_{A \rightarrow B}$ be a quantum channel from Alice to Bob. A CPTP map $\cM_{A \rar B}$ is a reverse Shannon simulation for $\mathcal{N}_{A \rightarrow B}$ if it consists of using a shared bipartite entangled state, applying local operation at Alice's side, sending $c$ classical bits from Alice to Bob, and at last applying local operation at Bob's side. The performance of the simulation is measured by the channel purified distance $P(\cM_{A\rar B},\cN_{A\rar B})$.
\end{definition}

We define the optimal performance of reverse Shannon simulation under the condition that the classical communication cost is bounded.
\begin{definition}[\em{performance function}]
Let $\mathcal{N}_{A \rightarrow B}$ be a quantum channel. For given $c \geq 0$, the optimal performance of reverse Shannon simulation for $\mathcal{N}_{A \rightarrow B}$, given that the classical communication cost is upper bounded by $c$ bits, is defined as
\[
P^{\rm{sim}}(\mathcal{N}_{A \rightarrow B}, c)
:= \min_{\mathcal{M}_{A \rightarrow B}\in\mathfrak{S}_c} P(\mathcal{M}_
{A \rightarrow B}, \mathcal{N}_{A \rightarrow B}),
\]
where $\mathfrak{S}_c$ is the set of all the reverse Shannon simulations of $\mathcal{N}_{A \rightarrow B}$ whose classical communication cost does not exceed $c$ bits.
\end{definition}

For $n$ copies of the channel $\mathcal{N}_{A \rightarrow B}$, the optimal performance $P^{\rm{sim}}(\mathcal{N}^{\ox n}, nr)$ is expected to decrease exponentially with $n$. The reliability
function of the reverse Shannon simulation of $\mathcal{N}_{A \rightarrow B}$ is defined as the rate of such exponential decreasing.

\begin{definition}[\em{reliability function}]
\label{def:reliability}
Let $\mathcal{N}_{A \rightarrow B}$ be a quantum channel. For given $r \geq 0$, the reliability function
$E^{\rm{sim}}(\cN_{A\rar B},r)$ of the reverse Shannon simulation of $\mathcal{N}_{A \rightarrow B}$ is defined as
\[
E^{\rm{sim}}(\cN_{A\rar B}, r)
:=\liminf_{n\rightarrow\infty} \frac{-1}{n} \log P^{\rm{sim}}\big(\mathcal{N}^{\ox n}_{A \rar B}, nr\big).
\]
\end{definition}

In fact, we can also use quantum communication and unlimited entanglement to simulate the channel $\mathcal{N}_{A \rightarrow B}$. This is called reverse Shannon simulation via quantum communication.
\begin{definition}
\label{definition:simuquantum}
Let $\mathcal{N}_{A \rightarrow B}$ be a quantum channel from Alice to Bob. A reverse Shannon simulation via quantum communication for $\mathcal{N}_{A \rightarrow B}$ consists of using a shared bipartite entangled state $\phi_{A'B'}$, applying unitary transformation $U_{AA' \rightarrow M\tilde{A}}$ at Alice's side, sending the system $M$ from Alice to Bob, and at last applying unitary transformation $V_{B'M \rightarrow B\tilde{B}}$ at Bob's side. The simulation map is given by
\[
\mathcal{M}_{A \rightarrow B}(\cdot):=\tr_{\tilde{A}\tilde{B}}
\Big[V_{B'M \rar B\tilde{B}}\circ U_{AA' \rar M\tilde{A}}\big((\cdot) \ox \phi_{A'B'}\big)\Big].
\]
\end{definition}

Similarly, we define the optimal performance under the condition that the quantum communication cost is bounded.
\begin{definition}
Let $\mathcal{N}_{A \rightarrow B}$ be a quantum channel. For given $c \geq 0$, the optimal performance of reverse Shannon simulation via quantum communication for $\mathcal{N}_{A \rightarrow B}$, given that the quantum communication cost is upper bounded by $c$ qubits, is defined as
\[
P_{\rm{q}}^{\rm{sim}}(\mathcal{N}_{A \rightarrow B}, c)
:= \min_{\mathcal{M}_{A \rightarrow B}\in\mathfrak{Q}_c} P(\mathcal{M}_
{A \rightarrow B}, \mathcal{N}_{A \rightarrow B}),
\]
where $\mathfrak{Q}_c$ is the set of all the reverse Shannon simulations of $\mathcal{N}_{A \rar B}$ with quantum communication not more than $c$ qubits.
\end{definition}

The following relation is a direct result of the teleportation and dense coding protocols~\cite{BBCJPW1993teleporting, BennettWiesner1992communication}.
\begin{lemma}
\label{lem:relation}
Let $\mathcal{N}_{A \rightarrow B}$ be a quantum channel. For any $c \geq 0$ we have
\[
P^{\rm{sim}}(\mathcal{N}_{A \rightarrow B},2c)=P_{\rm{q}}^{\rm{sim}}(\mathcal{N}_{A \rightarrow B},c).
\]
\end{lemma}

\section{Symmetrization and de Finetti Reduction}
\label{sec:symmetrizaion}
Naturally, a simulation channel for $\cN_{A\rar B}^{\ox n}$ has input system $A^n\equiv A_1A_2...A_n$ and output system $B^n\equiv B_1B_2...B_n$. Let $S_n$ denote the symmetric group over the set $\{1,2,\ldots,n\}$. Let $W^\pi_{A^n}$ be the natural unitary representation of $\pi\in S_n$ on $\cH_A^{\ox n}$, given by
\[
W^\pi_{A^n}: \ket{\psi_1} \ox \cdots \ox \ket{\psi_n} \mapsto
\ket{\psi_{\pi^{-1}(1)}} \ox \cdots \ox \ket{\psi_{\pi^{-1}(n)}}, \quad\forall\ \ket{\psi_i} \in \cH_A.
\]
We say that a channel $\cM_{A^n\rar B^n}$ is symmetric if
\[\cM_{A^n\rar B^n}=W_{B^n}^{\pi^{-1}}\circ\cM_{A^n\rar B^n}\circ W_{A^n}^\pi\]
holds for any permutation $\pi\in S_n$.

To achieve the best performance, we can always symmetrize the simulation channel $\cM_{A^n\rar B^n}$ by randomly applying a permutation and its inverse, respectively before and after the action of $\cM$. This does not cause any loss of performance, as the target channel $\cN_{A\rar B}^{\ox n}$ itself is symmetric. The common randomness needed in the symmetrization can be obtained by making measurements at both sides of some shared maximally entangled states.

\begin{lemma}\label{lem:symmetrization}
Given the channel $\cN_{A\rar B}$ and a simulation channel $\cM_{A^n\rar B^n}$ for $\cN_{A\rar B}^{\ox n}$, we let
\[
 \widetilde{\cM}_{A^n\rar B^n}
=\sum_{\pi\in S_n}\frac{1}{|S_n|}W_{B^n}^{\pi^{-1}}\circ\cM_{A^n\rar B^n}\circ W_{A^n}^\pi
\]
be the symmetrization of $\cM_{A^n\rar B^n}$. Then
\[P\big(\widetilde{\cM}_{A^n\rar B^n},\cN_{A\rar B}^{\ox n}\big)
\leq P\big(\cM_{A^n\rar B^n},\cN_{A\rar B}^{\ox n}\big).\]
\end{lemma}

\begin{IEEEproof}
For any pure state $\varphi_{RA^n}$, we have
\begin{align*}
     & F\big(\widetilde{\cM}(\varphi_{RA^n}),\cN^{\ox n}(\varphi_{RA^n})\big)  \\
  =  & F\Big(\sum_{\pi\in S_n}\frac{1}{|S_n|}W_{B^n}^{\pi^{-1}}\circ\cM\circ W_{A^n}^\pi(\varphi_{RA^n}),
       \sum_{\pi\in S_n}\frac{1}{|S_n|}W_{B^n}^{\pi^{-1}}\circ\cN^{\ox n}\circ
        W_{A^n}^\pi(\varphi_{RA^n})\Big) \\
\geq & \sum_{\pi\in S_n}\frac{1}{|S_n|} F\big(W_{B^n}^{\pi^{-1}}\circ\cM\circ W_{A^n}^\pi(\varphi_{RA^n}),
        W_{B^n}^{\pi^{-1}}\circ\cN^{\ox n}\circ W_{A^n}^\pi(\varphi_{RA^n})\big) \\
  =  & \sum_{\pi\in S_n}\frac{1}{|S_n|} F\big(\cM ( W_{A^n}^\pi(\varphi_{RA^n})),
       \cN^{\ox n} ( W_{A^n}^\pi(\varphi_{RA^n}))\big) \\
\geq & \sum_{\pi\in S_n}\frac{1}{|S_n|} \min_{\varphi'_{RA^n}}F\big(\cM (\varphi'_{RA^n}),
       \cN^{\ox n} (\varphi'_{RA^n})\big) \\
  =  & \min_{\varphi'_{RA^n}}F\big(\cM (\varphi'_{RA^n}),\cN^{\ox n} (\varphi'_{RA^n})\big),
\end{align*}
where the first equality is due to the fact that $\cN_{A\rar B}^{\ox n}$ is symmetric, and the first inequality follows from the joint concavity of the fidelity~\cite{Watrous2018the}. This is equivalent to
\[\begin{split}
      P\big(\widetilde{\cM}(\varphi_{RA^n}),\cN^{\ox n}(\varphi_{RA^n})\big)
\leq &\max_{\varphi'_{RA^n}}P\big(\cM (\varphi'_{RA^n}),\cN^{\ox n} (\varphi'_{RA^n})\big)\\
  =  &P\big(\cM_{A^n\rar B^n},\cN_{A\rar B}^{\ox n}\big).
\end{split}\]
As $\varphi_{RA^n}$ is chosen arbitrarily, the statement follows.
\end{IEEEproof}

\medskip
Now, once the simulation process is restricted to be symmetric, the technique of de Finetti reduction~\cite{CKR2009postselection} applies. It states that the simulation is good for any input states as long as it is good for the de Finetti state $\zeta_{R_nA^n}$, which is the purification of
\[
\zeta_{A^n}:=\int\rho_A^{\ox n}\di\,\mu(\rho)
\]
with $\mu(\cdot)$ the probability measure on the space of density operators on $\cH_A$ induced by the
Hilbert-Schmidt metric. An alternative understanding for $\mu(\cdot)$ is by introducing a purification
system $\bR\cong A$. Then $\mu(\rho)=\nu(\phi_{\bR A}(\rho))$, where $\phi_{\bR A}(\rho)$ is the purification of $\rho_A$ and $\nu(\cdot)$ is the uniform spherical probability measure satisfying $\nu(\phi)=\nu(u\phi u^*)$ for any unitary operator $u$ on $\cH_{\bR A}$.

The original de Finetti reduction of~\cite{CKR2009postselection} was derived under the diamond norm as the measure of distance between quantum channels. Here we reformulate this result by employing the purified distance. For completeness, we give a proof in the Appendix. We point out that this reformulation based on the purified distance is made independently in~\cite{RTB2023moderate} as well.
\begin{lemma}\label{lem:deFa}
Let $\cM_{A^n\rar B^n}$ and $\cE_{A^n\rar B^n}$ be symmetric quantum channels such that
$\cM=W_{B^n}^{\pi^{-1}}\circ\cM\circ W_{A^n}^\pi$ and $\cE=W_{B^n}^{\pi^{-1}}\circ\cE\circ W_{A^n}^\pi$ for any permutation $\pi\in S_n$. Set $g_{n,|A|}:=\tbinom{n+|A|^2-1}{n}\leq (n+1)^{|A|^2-1}$. We have
\[
P(\cM_{A^n\rar B^n}, \cE_{A^n\rar B^n})
\leq \sqrt{g_{n,|A|}}\ P\big(\cM(\zeta_{R_nA^n}),\cE(\zeta_{R_nA^n})\big).
\]
\end{lemma}

\section{Main Results}
\label{sec:main-results}
In this section, we present the main results of this paper. Their proofs will be given in the next two sections. The first result is an exponential achievability bound for the finite-blocklength setting, namely, for the simulation of $\cN_{A\rar B}^{\ox n}$ with any integer $n$.

\begin{theorem}
\label{thm:achievability}
Let $\cN_{A\rar B}$ be a quantum channel. For the reverse Shannon simulation of $\cN_{A \rar B}^{\ox n}$ with bounded classical communication rate $r\geq 0$, we have
\begin{equation*}
P^{\rm{sim}}\left(\cN^{\ox n}, nr\right)
\leq f(n,s)\exp{\Big\{\!-n\frac{s}{2}\big(r-I_{1+s}(\cN)\big)\Big\}},
\end{equation*}
where $s\in(0,1]$ is arbitrary, and $f(n,s):=\sqrt{s^{-1}(n+1)^{(1+s)(d+1)^2}}$ with $d=\max\{|A|,|B|\}$.
\end{theorem}

Then, we derive a one-shot converse bound. This is given by the smoothing quantity associated with the max-information of any bipartite state that can be generated by the channel. Let $\rho_{RB} \in \cS(RB)$ be a bipartite quantum state. Recall that for any real number $\lambda$ the smoothing quantity associated with the max-information of $\rho_{RB}$ is defined~\cite{LiYao2024reliability} as
\begin{equation}
\label{eq:smooths1}
\delta_{R:B}(\rho_{RB},\lambda):=
\min_{\sigma_B\in\cS(B)} \delta(\rho_{RB}\|\rho_R\ox\sigma_B,\lambda),
\end{equation}
where for $\varrho \in \cS(\cH)$ and $\omega \in \cP(\cH)$,
\begin{equation}
\label{eq:smooths2}
\delta(\varrho \| \omega,\lambda):=\min\left\{ P(\tilde{\varrho}, \varrho): \tilde{\varrho}\in
\cS_{\leq}(\cH),\,\tilde{\varrho} \leq 2^\lambda\omega\right\}.
\end{equation}
Employing this quantity, we obtain the following one-shot converse bound.
\begin{theorem}
  \label{thm:converse}
Let $\cN_{A\rar B}$ be a quantum channel and $c\geq 0$. For any bipartite pure state $\varphi_{RA}$ we have
\[
  P^{\rm{sim}}(\mathcal{N}, c) \geq \delta_{R:B}\big(\cN(\varphi_{RA}),c\big).
\]
\end{theorem}

Based on Theorem~\ref{thm:achievability} and Theorem~\ref{thm:converse}, we prove the following characterization of the reliability function.
\begin{theorem}
\label{thm:reliability}
Let $\cN_{A \rar B}$ be a quantum channel and $r \geq 0$. The reliability function $E^{\rm{sim}}(\cN,r)$ for the reverse Shannon simulation of $\cN_{A \rar B}$ satisfies
\begin{align}
E^{\rm{sim}}(\cN, r) &\geq \frac{1}{2} \max_{0 \leq s \leq 1}
\big\{s\big(r-I_{1+s}(\mathcal{N})\big)\big\},
  \label{eq:reliability-l}\\
E^{\rm{sim}}(\cN, r) &\leq \frac{1}{2} \sup_{s \geq 0}
\big\{s\big(r-I_{1+s}(\mathcal{N})\big)\big\}.
  \label{eq:reliability-u}
\end{align}
In particular, when $r\leq R_{\rm{critical}}
:=\frac{\mathrm{d}}{\mathrm{d}s}sI_{1+s}(\mathcal{N})\big|_{s=1}$,
\begin{equation}
  \label{eq:reliability-e}
E^{\rm{sim}}(\cN, r)=\frac{1}{2} \max_{0 \leq s \leq 1}
\big\{s\big(r-I_{1+s}(\mathcal{N})\big)\big\}.
\end{equation}
\end{theorem}

The results of Theorem~\ref{thm:reliability} are depicted in Fig.~\ref{fig:reliability}. As long as the classical communication rate $r$ is below the critical value $R_\text{critical}$, we have determined the explicit formula for the reliability function. When the classical communication rate $r$ is above the critical value, the two bounds diverge from each other more and more (cf. Proposition~\ref{prop:propertyreli}). It is worth pointing out that when $r\geq 2\log\min\{|A|,|B|\}$, perfect simulation can be achieved by teleportation~\cite{BBCJPW1993teleporting}, even in the one-shot setting. So in this situation the reliability function is infinity.

To see at what range of the communication rate we have obtained the exact error exponent, we computer the critical value and the mutual information of depolarizing channels; see Example~\ref{exam:crit-mi} in the Appendix. It shows that $R_\text{critical}$ is considerably larger than $I(\cN)$.
\begin{figure}[ht]
  \includegraphics[width=8cm]{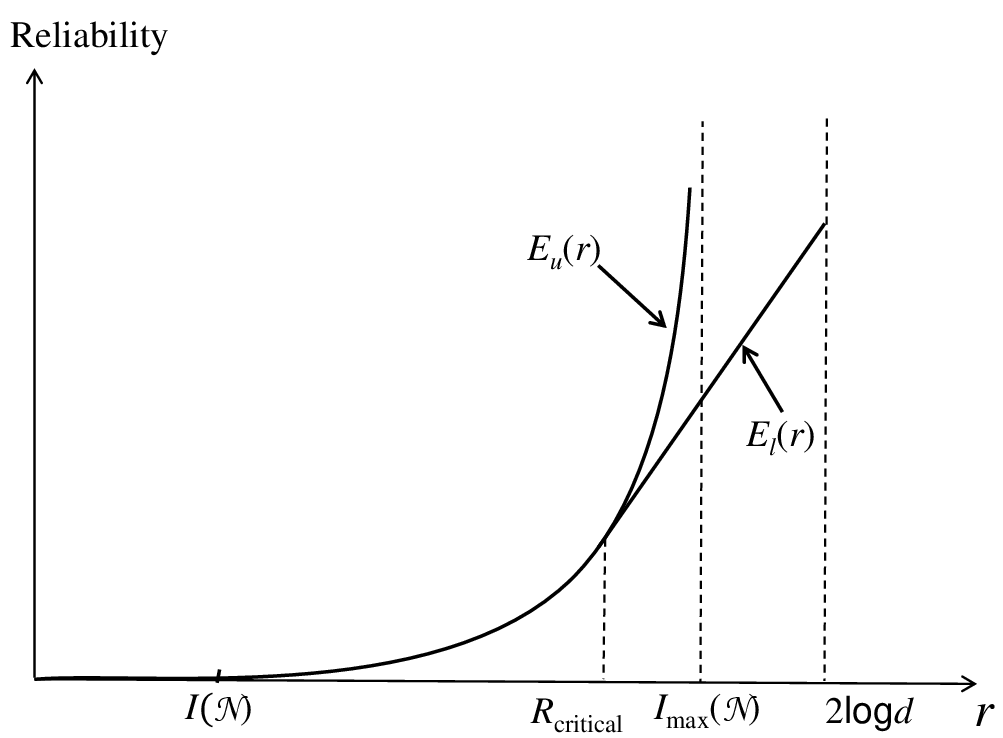} \centering
  \caption{Reliability function of quantum channel simulation.
  $E_l(r):=\frac{1}{2} \max_{0 \leq s \leq 1}\big\{s\big(r-I_{1+s}(\mathcal{N})\big)\big\}$ is the lower bound of Eq.~(\ref{eq:reliability-l}).
  $E_u(r):=\frac{1}{2} \sup_{s \geq 0}\big\{s\big(r-I_{1+s}(\mathcal{N})\big)\big\}$ is the upper bound of Eq.~(\ref{eq:reliability-u}).
  The two bounds are equal in the interval $[0,R_\text{critical}]$, giving the exact reliability function. Above the critical value $R_\text{critical}$, the upper bound $E_u(r)$ increases faster and it diverges to infinity when $r>I_\text{max}(\cN):=\lim\limits_{\alpha\rar\infty}I_\alpha(\cN)$, while the lower bound $E_l(r)$ becomes linear and reaches $\log d-\frac{1}{2}I_2(\cN)$ at $r=2\log d$ with $d:=\min\{|A|,|B|\}$.}
  \label{fig:reliability}
\end{figure}

The case where the classical communication rate $r$ is around the channel's mutual information $I(\cN)$ is of particular interest. We will see that the Quantum Reverse Shannon Theorem follows as an immediate corollary of Theorem~\ref{thm:reliability}.

\begin{corollary}[Quantum Reverse Shannon Theorem~\cite{BDHSW2014quantum,BCR2011the}]
\label{cor:QRST}
For a quantum channel $\cN_{A \rar B}$, the minimum classical communication rate $C_{\rm{QRST}}$ needed for asymptotically reliable simulation of $\cN_{A \rar B}$ is given by the channel's mutual information $I(\cN)$.
\end{corollary}

It is shown in~\cite{CMW2016strong} that $I_{1+s}(\cN)$ is an increasing function of $s$ and $I_{1+s}(\cN)$ converges to $I(\cN)$ as $s$ goes to $0$. With this we get from Theorem~\ref{thm:reliability} that $E^{\rm{sim}}(\cN, r)=0$ when $r \leq I(\cN)$, and $E^{\rm{sim}}(\cN, r)>0$ when $r > I(\cN)$. The latter implies that any rate $r$ above $I(\cN)$ is sufficient for asymptotically reliable simulation. So, $C_{\rm{QRST}} \leq I(\cN)$. On the other hand, $C_{\rm{QRST}} \geq I(\cN)$ is obvious~\cite{BDHSW2014quantum,BCR2011the}, because the rate of classical communication necessary for reliable simulation of $\cN_{A \rar B}$ has to be at least equal to the channel's entanglement-assisted classical capacity, which is $I(\cN)$~\cite{BSST2002entanglement}.

\section{Achievability Bounds}
\label{sec:Proofs-achi}
In this section, we prove Theorem~\ref{thm:achievability} as well as the achievability part of Theorem~\ref{thm:reliability}. At first, in Section~\ref{subsec:fixed-input} we consider the problem of channel simulation with a fixed input state and obtain a one-shot achievability bound. Then, combining this with the technique of de Finetti reduction, we prove in Section~\ref{subsec:finite-n} the finite-blocklength bound of Theorem~\ref{thm:achievability}, from which the achievability part of Theorem~\ref{thm:reliability} follows.

\subsection{Channel Simulation with a Fixed Input}
\label{subsec:fixed-input}
When the input sate is fixed, the quantum state splitting protocol provides us an explicit simulation scheme~\cite{ADHW2009mother,BDHSW2014quantum,BCR2011the}. This is the reverse procedure of quantum state merging, which in turn is related to quantum information decoupling~\cite{ADHW2009mother,ADJ2017quantum,MBDRC2017catalytic}. Employing the catalytic version of quantum information decoupling and its induced state merging strategy of~\cite{ADJ2017quantum, MBDRC2017catalytic}, we derive the following bound based on the analysis of~\cite{LiYao2024reliability}.

\begin{proposition}
\label{prop:achifixin}
Let $\cN_{A \rar B}$ be a quantum channel. Fix the input state $\varphi_{RA}$ and let $\varphi_{RB}=\cN (\varphi_{RA})$. For arbitrary $c \geq 0$ and state $\sigma_B$, there is a reverse Shannon simulation $\cM_{A \rar B}$ for $\cN_{A \rar B}$ with $c$ bits of classical communication, such that for any $s\in(0,1]$,
\begin{align}
     & P\big(\cM(\varphi_{RA}),\cN(\varphi_{RA})\big) \nonumber\\
\leq & \sqrt{\frac{v^s}{s}}\exp\Big\{-\frac{s}{2} \big(c-D_{1+s}
       (\varphi_{RB}\| \varphi_R\ox\sigma_B)\big)\Big\}, \nonumber
\end{align}
where $v$ is the number of distinct eigenvalues of $\varphi_R\ox \sigma_B$.
\end{proposition}

\begin{IEEEproof}
Let $U_\cN^{A\rar BE}$ be the Stinespring dilation of the channel $\cN_{A\rar B}$.
We start with an introduction of quantum state merging with respect to the state $\varphi_{REB}=\id_R\ox U_\cN^{A\rar BE}(\varphi_{RA})$. Suppose that $\varphi_{REB}$ is held by a referee ($R$), Alice ($E$), and Bob ($B$). This is compactly represented by  $\varphi_{R:E:B}$. That is, when multiple systems are separated by two colons, it means that the three parts from left to right are respectively held by the referee, Alice and Bob. Based on the protocol introduced in~\cite{MBDRC2017catalytic, ADJ2017quantum}, a tight achievability bound has been derived in~\cite{LiYao2024reliability}, for merging the $B$ part of $\varphi_{R:E:B}$ to Alice. The merging scheme consists of the following steps:
\begin{enumerate}[~~~~(a$'$)]
  \item Alice and Bob share a pair of entangled pure state $\gamma_{A'B'}$ (here $\gamma_{A'B'}$ depends
        on $\sigma_B$, and it can be denoted as $\gamma_{\emptyset:A':B'}$, where $\emptyset$ stands for an empty set of systems);
  \item Bob applies a local unitary transformation $V_{BB'\rar B_1B_2}$, where the size of the the $B_2$
        system is $\log|B_2|=\frac{c}{2}$;
  \item Bob sends the $B_2$ system to Alice via the noiseless channel $\id_{B_2\rar \bB_2}$ with $B_2
        \cong \bB_2$;
  \item Alice applies a local unitary transformation $W_{EA'\bB_2\rar BEA_1}$.
\end{enumerate}
Write the resulting state as $\Psi_{R:BEA_1:B_1}=W\circ\id_{B_2\rar \bB_2}\circ V (\varphi_{R:E:B} \ox \gamma_{\emptyset:A':B'})$. Then there exists a pure entangled state $\omega_{A_1B_1}$ shared between Alice and Bob, such that~\cite[Theorem 7 and Proposition 11]{LiYao2024reliability}
\begin{equation}
     P(\Psi_{R:BEA_1:B_1},\varphi_{R:BE:\emptyset}\ox\omega_{\emptyset:A_1:B_1})
\leq \sqrt{\frac{v^s}{s}}\exp\Big\{-\frac{s}{2} \big(c-D_{1+s}
       (\varphi_{RB}\| \varphi_R\ox\sigma_B)\big)\Big\}, \label{eq:mergingb-a}
\end{equation}
where $s$ and $v$ are as specified in the statement of the proposition.

Now, for the input state $\varphi_{RA}$, the reverse Shannon simulation for $\cN_{A\rar B}$ is constructed essentially by reversing the above steps (this is called state splitting):
\begin{enumerate}[~~~~(a)]
  \item Alice and Bob share the entangled state $\omega_{A_1B_1}$ for assistance;
  \item Alice applies the Stinespring dilation $U_\cN^{A\rar BE}$ locally (now the whole state is
        $\varphi_{R:BE:\emptyset}\ox\omega_{\emptyset:A_1:B_1}$);
  \item Alice applies the unitary transformation $W^{-1}$;
  \item Alice implements the noiseless transmission $\id_{{\bB}_2 \rar B_2}$ to Bob via quantum
        teleportation, which consumes $c$ bits of classical communication;
  \item Bob applies the unitary transformation $V^{-1}$.
\end{enumerate}
We denote the resulting state as $\Phi_{R:EA':BB'}$. Then
\[
\Phi_{R:EA':BB'}=V^{-1}\circ\id_{\bB_2\rar B_2}\circ W^{-1}
(\varphi_{R:BE:\emptyset} \ox \omega_{\emptyset:A_1:B_1}),
\]
and we have
\begin{align}
  & P\big(\Phi_{R:EA':BB'}, \varphi_{R:E:B}\ox \gamma_{\emptyset:A':B'}\big) \nonumber\\
= & P\big(V^{-1}\circ\id_{\bB_2\rar B_2}\circ W^{-1}(\varphi_{R:BE:\emptyset} \ox
      \omega_{\emptyset:A_1:B_1}),\varphi_{R:E:B}\ox \gamma_{\emptyset:A':B'}\big) \nonumber\\
= & P\big(\varphi_{R:BE:\emptyset} \ox \omega_{\emptyset:A_1:B_1},
      W\circ\id_{B_2\rar \bB_2}\circ V (\varphi_{R:E:B} \ox \gamma_{\emptyset:A':B'}\big) \nonumber\\
= & P\big(\Psi_{R:BEA_1:B_1},\varphi_{R:BE:\emptyset} \ox \omega_{\emptyset:A_1:B_1}\big),
     \label{eq:mergingb-b}
\end{align}
where for the second equality we apply the unitary operation $W\circ\id_{B_2\rar \bB_2}\circ V$ to both states, and for the last equality we exchange the position of the two states. At last, the simulation map $\cM_{A\rar B}$ is obtained by discarding the systems $E$, $A'$ and $B'$ after the above steps (a)-(e). Thus,
\[
\cM_{A\rar B}(\varphi_{RA}) = \tr_{EA'B'}[\Phi_{R:EA':BB'}].
\]
So, we have
\begin{align*}
     & P\big(\cM(\varphi_{RA}),\cN(\varphi_{RA})\big) \\
  =  & P\big(\tr_{EA'B'}[\Phi_{R:EA':BB'}], \varphi_{R:\emptyset:B}\big) \\
\leq & P\big(\Phi_{R:EA':BB'}, \varphi_{R:E:B}\ox \gamma_{\emptyset:A':B'}\big) \\
\leq & \sqrt{\frac{v^s}{s}}\exp\Big\{-\frac{s}{2} \big(c-D_{1+s}
       (\varphi_{RB}\| \varphi_R\ox\sigma_B)\big)\Big\},
\end{align*}
where the first inequality is by the monotonicity of purified distance under partial trace, and the second inequality comes from the combination of Eq.~\eqref{eq:mergingb-a} and Eq.~\eqref{eq:mergingb-b}.
\end{IEEEproof}

\subsection{Finite-Blocklength Achievability Bound and Proof of Eq.~(\ref{eq:reliability-l})}
\label{subsec:finite-n}
We complete the proof of Theorem~\ref{thm:achievability} and Eq.~(\ref{eq:reliability-l}) of Theorem~\ref{thm:reliability} in this subsection. At first, we show the following lemma.
\begin{lemma}\label{lem:minimizer}
For the quantum state $\zeta_{R_nB^n}=\cN_{A\rar B}^{\ox n}(\zeta_{R_nA^n})$ with $\zeta_{R_nA^n}$ being the de Finetti state of Lemma~\ref{lem:deFa}, and for $s\geq-\frac{1}{2}$, let $\sigma^*_{B^n}$ be a minimizer of
\[
\min_{\sigma_{B^n}}D_{1+s}(\zeta_{R_nB^n}\| \zeta_{R_n}\ox\sigma_{B^n})=I_{1+s}(R_n:B^n)_{\zeta_{R_nB^n}}.
\]
Then $\sigma^*_{B^n}$ can be chosen to be symmetric, i.e., $\sigma^*_{B^n}=W^\pi_{B^n}(\sigma^*_{B^n})$ for all $\pi\in S_n$.
\end{lemma}

\begin{IEEEproof}
Because $\zeta_{A^n}$ is invariant under permutations, for any $\pi\in S_n$ the pure states $W^\pi_{A^n}(\zeta_{R_nA^n})$ and $\zeta_{R_nA^n}$ are identical on the $A^n$ system. So, there is a unitary operation $V_{R_n}$ such that
\begin{equation}\label{eq:minimizer-1}
V_{R_n}\ox W^\pi_{A^n}(\zeta_{R_nA^n})=\zeta_{R_nA^n}.
\end{equation}
Applying the channel $\cN_{A\rar B}^{\ox n}$ to both sides of Eq.~\eqref{eq:minimizer-1}, due to the symmetry that $W^{\pi}_{B^n}\circ\cN_{A\rar B}^{\ox n}=\cN_{A\rar B}^{\ox n}\circ W^\pi_{A^n}$, we get
\begin{equation}\label{eq:minimizer-2}
V_{R_n}\ox W^\pi_{B^n}(\zeta_{R_nB^n})=\zeta_{R_nB^n}.
\end{equation}
Eq.~\eqref{eq:minimizer-2} ensures that we also have $V_{R_n}(\zeta_{R_n})=\zeta_{R_n}$. With these two relations we make use of the property that the sandwiched R\'enyi divergence is invariant under unitary operations to obtain, for any permutation $\pi$ and quantum state $\sigma_{B^n}$,
\begin{align}
 &D_{1+s}\left(\zeta_{R_nB^n}\| \zeta_{R_n}\ox\sigma_{B^n}\right) \nb \\
=&D_{1+s}\big(V_{R_n}\ox W^\pi_{B^n}(\zeta_{R_nB^n})\| V_{R_n}(\zeta_{R_n})\ox
  W^\pi_{B^n}(\sigma_{B^n})\big) \nb\\
=&D_{1+s}\big(\zeta_{R_nB^n}\| \zeta_{R_n}\ox W^\pi_{B^n}(\sigma_{B^n})\big). \label{eq:minimizer-3}
\end{align}
Let $\tilde{\sigma}_{B^n}=\sum_{\pi\in S_n}\frac{1}{|S_n|}W^\pi_{B^n}(\sigma_{B^n})$ be the symmetrization of $\sigma_{B^n}$. Then
\begin{align}
     &D_{1+s}\left(\zeta_{R_nB^n}\| \zeta_{R_n}\ox\tilde{\sigma}_{B^n}\right) \nb \\
\leq &\sum_{\pi\in S_n}\frac{1}{|S_n|}D_{1+s}\big(\zeta_{R_nB^n}\| \zeta_{R_n}\ox
             W^\pi_{B^n}(\sigma_{B^n})\big) \nb\\
  =  &\sum_{\pi\in S_n}\frac{1}{|S_n|}D_{1+s}\left(\zeta_{R_nB^n}\| \zeta_{R_n}\ox\sigma_{B^n}\right)\nb \\ =  &D_{1+s}\left(\zeta_{R_nB^n}\| \zeta_{R_n}\ox\sigma_{B^n}\right), \label{eq:minimizer-4}
\end{align}
where the inequality is by the convexity of the function $\omega\mapsto D_{1+s}(\rho\|\omega)$ for $s\geq-\frac{1}{2}$ and $\rho,\omega\in\mathcal{S}(\mathcal{H})$~\cite{MosonyiOgawa2017strong}, and the first equality is by Eq.~\eqref{eq:minimizer-3}. Eq.~\eqref{eq:minimizer-4} implies the statement and we are done.
\end{IEEEproof}

\bigskip
\begin{IEEEproof}[Proof of Theorem~\ref{thm:achievability}]
Since we can always symmetrize the simulation channel, according to Lemma~\ref{lem:deFa}, it suffices to consider the de Finetti input state $\zeta_{R_nA^n}$. Let $\zeta_{R_nB^n}=\cN_{A\rar B}^{\ox n}(\zeta_{R_nA^n})$. From now on, we fix an $s\in(0,1]$. As shown in Lemma~\ref{lem:minimizer}, we can choose a symmetric (permutation-invariant) state $\sigma^*_{B^n}$ to minimize $D_{1+s}(\zeta_{R_nB^n}\| \zeta_{R_n}\ox\sigma_{B^n})$, such that
\begin{equation} \label{eq:achib-1}
D_{1+s}(\zeta_{R_nB^n}\| \zeta_{R_n}\ox\sigma^*_{B^n})
=I_{1+s}(R_n:B^n)_{\cN_{A\rar B}^{\ox n}(\zeta_{R_nA^n})}.
\end{equation}
Denote the number of distinct eigenvalues of a quantum state $\rho$ as $|\operatorname{spec}(\rho)|$, where $\operatorname{spec}(\cdot)$ stands for the spectral set. For a symmetric state $\rho_{X^n}$, $|\operatorname{spec}(\rho_{X^n})|$ is polynomial in $n$, upper bounded by the sum of the dimensions of all the irreducible decompositions of the representation $u\mapsto u^{\ox n}$ of the unitary group on  $\mathcal{H}_X^{\ox n}$. An explicit bound of $(n+1)^{\frac{(|X|+2)(|X|-1)}{2}}$ can be found in, e.g.,~\cite{Hayashi2010universal}. So, we have
\begin{equation*}
|\operatorname{spec}(\zeta_{R_n}\ox\sigma^*_{B^n})|
\leq |\operatorname{spec}(\zeta_{A^n})|\cdot |\operatorname{spec}(\sigma^*_{B^n})|
\leq (n+1)^{(d+2)(d-1)},
\end{equation*}
where $d=\max\{|A|,|B|\}$, and for the first inequality we have used $|\operatorname{spec}(\zeta_{R_n})|
=|\operatorname{spec}(\zeta_{A^n})|$ since $\zeta_{R_nA^n}$ is a pure state.

Now, we apply Proposition~\ref{prop:achifixin} to the channel $\cN_{A\rar B}^{\ox n}$ with the fixed input state $\zeta_{R_nA^n}$. Set $c=nr$ and $\sigma_{B^n}=\sigma^*_{B^n}$. According to Proposition~\ref{prop:achifixin}, there is a reverse Shannon simulation ${\cM'}_{A^n\rar B^n} $ with $c=nr$ bits of classical communication, such that for $s$ fixed earlier,
\begin{align}
     &P\left({\cM'}_{A^n\rar B^n} (\zeta_{R_nA^n}),\cN_{A\rar B}^{\ox n}(\zeta_{R_nA^n})\right)\nb\\
\leq &\sqrt{\frac{(n+1)^{(d+2)(d-1)s}}{s}}\exp\Big\{-\frac{s}{2} \big(nr-D_{1+s}
      (\zeta_{R_nB^n}\| \zeta_{R_n}\ox\sigma^*_{B^n})\big)\Big\} \nb\\
  =  &\sqrt{\frac{(n+1)^{(d+2)(d-1)s}}{s}}\exp\Big\{-\frac{s}{2} \big(nr
      -I_{1+s}(R_n:B^n)_{\cN_{A\rar B}^{\ox n}(\zeta_{R_nA^n})}\big)\Big\}, \label{eq:achib-3}
\end{align}
where the equality is by Eq.~\eqref{eq:achib-1}. Next, we symmetrize ${\cM'}_{A^n\rar B^n} $ to obtain a symmetric simulation $\cM_{A^n\rar B^n} $. This is the final simulation scheme that we need. The symmetrization consumes some classical randomness which can be obtained by making measurements on the shared entanglement. Meanwhile, it does not cause loss of performance, which can be verified by an argument similar to the proof of Lemma~\ref{lem:symmetrization}:
\begin{align}
     &P\left({\cM}_{A^n\rar B^n} (\zeta_{R_nA^n}),\cN_{A\rar B}^{\ox n}(\zeta_{R_nA^n})\right)\nb\\
  =  &P\Big(\sum_{\pi\in S_n}\frac{1}{|S_n|}W_{B^n}^{\pi^{-1}}\circ{\cM'}_{A^n\rar B^n} \circ
       W_{A^n}^\pi(\zeta_{R_nA^n}),\cN_{A\rar B}^{\ox n}(\zeta_{R_nA^n})\Big)\nb\\
\leq &\max_{\pi\in S_n}P\left(W_{B^n}^{\pi^{-1}}\circ{\cM'}_{A^n\rar B^n} \circ
       W_{A^n}^\pi(\zeta_{R_nA^n}),\cN_{A\rar B}^{\ox n}(\zeta_{R_nA^n})\right)\nb\\
  =  &P\left({\cM'}_{A^n\rar B^n} (\zeta_{R_nA^n}),\cN_{A\rar B}^{\ox n}(\zeta_{R_nA^n})\right),
       \label{eq:achib-4}
\end{align}
where the inequality follows from the concavity of the fidelity together with the definition of purified distance, and the last line is because both $\zeta_{A^n}$ and $\cN_{A\rar B}^{\ox n}$ are invariant under permutation and we also use a trick of Eq.~\eqref{eq:minimizer-1}. At this stage, we are ready to apply the de Finetti reduction of Lemma~\ref{lem:deFa} to get
\begin{align}
      P^{\rm{sim}}\left(\cN^{\ox n}, nr\right)
\leq &P\left({\cM}_{A^n\rar B^n} ,\cN_{A\rar B}^{\ox n}\right)\nb\\
\leq &\sqrt{g_{n,|A|}}P\left({\cM}_{A^n\rar B^n} (\zeta_{R_nA^n}),
                       \cN_{A\rar B}^{\ox n}(\zeta_{R_nA^n})\right). \label{eq:achib-5}
\end{align}
Combining Eqs.~(\ref{eq:achib-3})--(\ref{eq:achib-5}), we arrive at
\begin{equation}
P^{\rm{sim}}\left(\cN^{\ox n}, nr\right)
\leq \sqrt{\frac{(n+1)^{(1+s)(d+1)^2}}{s}}\exp\Big\{-\frac{s}{2}
\big(nr-I_{1+s}(R_n:B^n)_{\cN_{A\rar B}^{\ox n}(\zeta_{R_nA^n})}\big)\Big\}.\label{eq:achib-6}
\end{equation}
Eq.~\eqref{eq:achib-6} holds for any $s\in(0,1]$, because it was chosen arbitrarily at the beginning of the proof.

The last step is to employ the additivity property of the channel's sandwiched R\'enyi information established in~\cite{GuptaWilde2015multiplicativity}: for $\alpha\geq 1$ and any two channels $\cE_1$ and $\cE_2$,
$I_\alpha(\cE_1\ox\cE_2)=I_\alpha(\cE_1)+I_\alpha(\cE_2)$.
This yields
\begin{equation}\label{eq:achib-7}
I_{1+s}(R_n:B^n)_{\cN_{A\rar B}^{\ox n}(\zeta_{R_nA^n})}
\leq I_{1+s}(\cN_{A\rar B}^{\ox n})
 =   nI_{1+s}(\cN_{A\rar B}).
\end{equation}
Inserting Eq.~\eqref{eq:achib-7} into Eq.~\eqref{eq:achib-6} lets us complete the proof.
\end{IEEEproof}

\bigskip
\begin{IEEEproof}[Proof of Theorem~\ref{thm:reliability}: Eq.~(\ref{eq:reliability-l})]
It is a straightforward consequence of Theorem~\ref{thm:achievability} that
\[
\liminf_{n\rar\infty} \frac{-1}{n} \log P^{\rm{sim}}\big(\cN_{A \rar B}^{\ox n}, nr\big)
\geq \frac{1}{2} \big\{ s\big(r-I_{1+s}(N_{A\rar B})\big) \big\}
\]
holds for any $s\in[0,1]$. Maximizing the right hand side over $s\in[0,1]$ yields
\[E^{\rm{sim}}(\cN_{A \rar B}, r) \geq \frac{1}{2} \max_{0 \leq s \leq 1}
\big\{s\big(r-I_{1+s}(\mathcal{N}_{A \rightarrow B})\big)\big\},
\]
which is Eq.~(\ref{eq:reliability-l}).
\end{IEEEproof}

\section{Converse Bounds}
\label{sec:Proofs-converse}
In this section, we deal with the converse part. At first, in Section~\ref{subsec:one-shot-conv} we prove the one-shot converse bound of Theorem~\ref{thm:converse}. Then, based on this we finish the proof of Theorem~\ref{thm:reliability} in Section~\ref{subsec:reliability-conv}.

\subsection{One-Shot Converse Bound}
\label{subsec:one-shot-conv}
Inspecting Eq.~\eqref{eq:smooths2}, we can see that if $\omega$ is a normalized state and $\lambda\geq 0$, then $\tilde{\varrho}$ can be restricted to be a normalized state. This implies that the converse bound of Theorem~\ref{thm:converse} can be written as
\begin{equation} \label{eq:smooths0}
\delta_{R:B}(\cN(\varphi_{RA}),c)
=\min \left\{P(\tilde{\rho}_{RB},\cN(\varphi_{RA})):\tilde{\rho}_{RB}\in\cS(RB),\sigma_B\in\cS(B),
\tilde{\rho}_{RB}\leq 2^c\varphi_R\ox\sigma_B\right\}.
\end{equation}
In the following, we show that $\delta_{R:B}(\cN(\varphi_{RA}),c)$ is a lower bound of the optimal performance for simulating $\cN_{A\rar B}$.

\begin{lemma}
  \label{lemma:converse-a}
Let $\cN_{A\rar B}$ be a quantum channel and $c \geq 0$ be a real number. For any bipartite state $\varphi_{RA}$ we have
\begin{align*}
  P^{\rm{sim}}(\mathcal{N}_{A \rightarrow B}, c) \geq \min_{\cM\in\mathfrak{S}_c}P(\cM(\varphi_{RA}),\cN(\varphi_{RA})),
\end{align*}
where $\mathfrak{S}_c$ is the set of all the reverse Shannon simulations for $\mathcal{N}_{A \rightarrow B}$ whose classical communication cost does not exceed $c$ bits.
\end{lemma}

\begin{IEEEproof}
It is straightforward to check that
\begin{align*}
      P^{\rm{sim}}(\mathcal{N}_{A \rightarrow B}, c)
  =  &\min_{\cM\in\mathfrak{S}_c} P(\cM_{A\rar B}, \cN_{A\rar B}) \\
  =  &\min_{\cM\in\mathfrak{S}_c} \max_{\varphi'_{RA}} P(\cM(\varphi'_{RA}),\cN(\varphi'_{RA})) \\
\geq &\max_{\varphi'_{RA}}\min_{\cM\in\mathfrak{S}_c} P(\cM(\varphi'_{RA}),\cN(\varphi'_{RA})) \\
\geq &\min_{\cM\in\mathfrak{S}_c} P(\cM(\varphi_{RA}),\cN(\varphi_{RA})),
\end{align*}
where the two equalities are by definitions.
\end{IEEEproof}

\begin{lemma}
  \label{lemma:converse-b}
Let $\cN_{A\rar B}$ be a quantum channel and $c\geq 0$ be a real number. For any reverse Shannon simulation $\cM_{A\rar B}$ for $\cN_{A\rar B}$ with classical communication not more than $c$ bits, and any bipartite state $\varphi_{RA}$ we have
\begin{align*}
  P(\cM(\varphi_{RA}),\cN(\varphi_{RA})) \geq \delta_{R:B}(\cN(\varphi_{RA}),c).
\end{align*}
\end{lemma}

\begin{IEEEproof}
Let $\cM_{A\rar B}$ be a general reverse Shannon simulation with not more than $c$ bits of classical communication. It consists of (a) using a shared bipartite entangled state $\omega_{A_1B_1}$ for assistance, (b) applying local operation $\cE_{AA_1\rar X}$ with $\log |X|\leq c$ at Alice's side, (c) sending $X$ from Alice to Bob via a noiseless classical channel $\overline{\id}_{X\rar X}:\ \rho_X\mapsto \sum_x\bra{x}\rho_X\ket{x}\proj{x}_X$ with $\{\ket{x}\}_x$ a standard basis of $\cH_X$, (d) applying local operation $\cD_{XB_1\rar B}$ at Bob's side.

Given the input state $\varphi_{RA}$, we set
\[
\Phi_{RXB_1}:=\overline{\id}_{X\rar X}\circ\cE_{AA_1\rar X}(\varphi_{RA}\ox\omega_{A_1B_1}),
\]
which is the state after the implementing of steps (a), (b) and (c). Here the systems $X$ and $B_1$ are held by Bob. The state $\Phi_{RXB_1}$ is classical in $X$, i.e., it can be written in the following form:
\[
\Phi_{RXB_1}=\sum_{x}\proj{x}_X\ox p_x \Phi_{RB_1}^x.
\]
Therefore,
\[
\Phi_{RXB_1}\leq\sum_{x}\proj{x}_X\ox \big(\sum_xp_x \Phi_{RB_1}^x\big)=\1_X\ox\Phi_{RB_1}.
\]
Moreover, it is obvious that $\Phi_{RB_1}=\varphi_R\ox\omega_{B_1}$. So, we have
\[
\Phi_{RXB_1}\leq 2^c\varphi_R\ox\big(\frac{\1_X}{|X|}\ox\omega_{B_1}\big).
\]
The above inequality then ensures that the resulting state of the simulation satisfies
\begin{align}
  \cM_{A\rar B}(\varphi_{RA})=&\cD_{XB_1\rar B}(\Phi_{RXB_1}) \nonumber \\
                         \leq &2^c\varphi_R\ox \cD_{XB_1\rar B}\big(\frac{\1_X}{|X|}\ox\omega_{B_1}\big).
                         \label{eq:converse-p1}
\end{align}
Since $\cD_{XB_1\rar B}$ is a CPTP map, $\cD_{XB_1\rar B}\big(\frac{\1_X}{|X|}\ox\omega_{B_1}\big)$ is a normalized state on system $B$. At last, by the definition of $\delta_{R:B}(\cN(\varphi_{RA}),c)$ (cf. Eq.~\eqref{eq:smooths0}), we easily see that Eq.~(\ref{eq:converse-p1}) implies
\[
P(\cM(\varphi_{RA}),\cN(\varphi_{RA})) \geq \delta_{R:B}\left(\cN(\varphi_{RA}),c\right),
\]
concluding the proof.
\end{IEEEproof}

\bigskip
\begin{IEEEproof}[Proof of Theorem~\ref{thm:converse}]
The proof is done by combining Lemma~\ref{lemma:converse-a} and Lemma~\ref{lemma:converse-b}.
\end{IEEEproof}

\subsection{Final Proof of the Reliability Function}
\label{subsec:reliability-conv}
We are now in a position to finish the proof of Theorem~\ref{thm:reliability}. For convenience, we define for any density matrix $\rho_A\in\cS(A)$,
\begin{equation*}
I_\alpha(\cN_{A\rar B},\rho_A):=I_\alpha(R:B)_{\cN(\rho_{RA})},
\end{equation*}
where $\rho_{RA}$ is any purification of $\rho_A$. This definition is reasonable because any purification of $\rho_A$ gives the same value for $I_\alpha(R:B)_{\cN(\rho_{RA})}$. The following Lemma~\ref{lem:concavity} will be crucial for our derivation.

\begin{lemma}
  \label{lem:concavity}
Let $\cN_{A\rar B}$ be a quantum channel. For any $\alpha>1$, the function $\rho_A \mapsto I_\alpha(\cN_{A\rar B},\rho_A)$ is concave on the set $\cS(A)$ of density matrices on $\cH_A$.
Furthermore, for any $\alpha\geq\frac{1}{2}$, it is continuous.
\end{lemma}
\begin{IEEEproof}
At first, we prove the concavity. Let $\rho_A, \sigma_A\in\cS(A)$ be density matrices, and let $\psi_{RA}$ and $\phi_{RA}$ be the purifications of $\rho_A$ and $\sigma_A$, respectively. For any $0\leq\lambda\leq1$, consider the convex combination $\lambda\rho_A+(1-\lambda)\sigma_A$. By introducing another system $F$, its purification can be written as
\[
\ket{\varphi}_{FRA}=\sqrt{\lambda}\ket{0}_F\ket{\psi}_{RA}+\sqrt{1-\lambda}\ket{1}_F\ket{\phi}_{RA}.
\]
Let $\cE_F:X\mapsto\bra{0}X\ket{0}\proj{0}_F+\bra{1}X\ket{1}\proj{1}_F$ be the measurement map acting on the system $F$. Then
\[
\cE_F(\varphi_{FRA})=\lambda\proj{0}_F\ox\psi_{RA}+(1-\lambda)\proj{1}_F\ox\phi_{RA}.
\]
Define for operators $X,Y\geq0$, $Q_\alpha(X\|Y):=\tr\big(Y^{\frac{1-\alpha}{2\alpha}}XY^{\frac{1-\alpha}{2\alpha}}\big)^\alpha$. We have
\begin{align*}
     &I_\alpha(\cN_{A\rar B},\lambda\rho_A+(1-\lambda)\sigma_A) \\
  =  &I_\alpha(FR:B)_{\cN_{A\rar B}(\varphi_{FRA})}             \\
  =  &\min_{\omega_B}D_\alpha\big(\cN_{A\rar B}(\varphi_{FRA})\|\varphi_{FR}\ox \omega_B\big) \\
\geq &\min_{\omega_B}D_\alpha\big(\cE_F\ox\cN_{A\rar B}(\varphi_{FRA})\|\cE_F(\varphi_{FR})\ox
                                                                               \omega_B\big) \\
  =  &\min_{\omega_B}D_\alpha\big(\lambda\proj{0}_F\ox\cN(\psi_{RA})+(1-\lambda)\proj{1}_F\ox\cN(\phi_{RA})
            \|                                                                               \\
     &\qquad\quad\ \ \lambda\proj{0}_F\ox\psi_R\ox\omega_B+(1-\lambda)\proj{1}_F\ox\phi_R\ox\omega_B\big)\\
  =  &\min_{\omega_B}\frac{1}{\alpha-1}\log\big(\lambda Q_\alpha(\cN(\psi_{RA})\|\psi_R\ox\omega_B) +
                                            (1-\lambda)Q_\alpha(\cN(\phi_{RA})\|\phi_R\ox\omega_B) \big) \\
\geq &\min_{\omega_B} \frac{1}{\alpha-1}\big\{\lambda\log Q_\alpha(\cN(\psi_{RA})\|\psi_R\ox\omega_B)
                                  +(1-\lambda)\log Q_\alpha(\cN(\phi_{RA})\|\phi_R\ox\omega_B)\big\} \\
  =  &\min_{\omega_B}\big\{\lambda D_\alpha\big(\cN(\psi_{RA})\|\psi_R\ox\omega_B\big)
                          +(1-\lambda)D_\alpha\big(\cN(\phi_{RA})\|\phi_R\ox\omega_B\big)\big\} \\
\geq &\lambda \min_{\omega_B} D_\alpha\big(\cN(\psi_{RA})\|\psi_R\ox\omega_B\big)
                          +(1-\lambda)\min_{\omega_B} D_\alpha\big(\cN(\phi_{RA})\|\phi_R\ox\omega_B\big)\\
  =  &\lambda I_\alpha(\cN_{A\rar B},\rho_A) + (1-\lambda) I_\alpha(\cN_{A\rar B},\sigma_A),
\end{align*}
where the first inequality is by the data processing inequality of the sandwiched R\'enyi divergence~\cite{Beigi2013sandwiched,FrankLieb2013monotonicity}.

Now, we show the continuity. Assume that $R\cong A$, with orthonormal bases $\{\ket{i}_R\}$ and $\{\ket{i}_A\}$, respectively. Let $\Phi_{RA}=\sum_{i,j=1}^{|A|}\ket{i}\bra{j}_R\ox\ket{i}\bra{j}_A$. Then $\sqrt{\rho_R}^\cT\Phi_{RA}\sqrt{\rho_R}^\cT$ is a purification of $\rho_A$, where ``$\cT$'' is the transpose operation. Denote by $\cS_+(B)$ the set of all positive states on $B$. We have
\[
I_{\alpha}(\cN_{A \rightarrow B}, \rho_A)
=\inf_{\sigma_B \in \cS_+(B)}
D_{\alpha}\big(\sqrt{\rho_R}^\cT\cN(\Phi_{RA})\sqrt{\rho_R}^\cT \| \rho_R^\cT \ox \sigma_B \big)
\]
and $D_{\alpha}\big(\sqrt{\rho_R}^\cT\cN(\Phi_{RA})\sqrt{\rho_R}^\cT \| \rho_R^\cT \ox \sigma_B \big)$ is
continuous in $\rho$ for any $\sigma_B \in \cS_+(B)$. Hence, $I_{\alpha}(\cN_{A \rightarrow B}, \rho_A)$ is upper semi-continuous in $\rho$. The remaining is to prove the lower semi-continuity.
Let $V_{A \rightarrow BE}$ be the Stinespring dilation for $\cN_{A \rightarrow B}$ and
$\cN^c_{A \rightarrow E}:\rho_A\mapsto\tr_B[V_{A \rightarrow BE}(\rho_A)V_{A \rightarrow BE}^*]$ be the complementary channel of $\cN_{A \rightarrow B}$. The duality relation for the sandwiched R\'enyi mutual information~\cite{HayashiTomamichel2016correlation} gives that
\begin{equation}\label{eq:conti}
I_{\alpha}(\cN_{A \rightarrow B}, \rho_A)
=\sup_{\sigma_E\in \cS_+(E)}
-D_{\beta}\big( \sqrt{\rho_R}^\cT\cN^c(\Phi_{RA})\sqrt{\rho_R}^\cT\| (\rho_R^\cT)^{-1} \ox \sigma_E \big),
\end{equation}
where $\frac{1}{\alpha}+\frac{1}{\beta}=2$ and $(\rho_R^\cT)^{-1}$ is the generalized inverse of $\rho_R^\cT$. Since $D_{\beta}( \sqrt{\rho_R}^\cT\cN^c(\Phi_{RA})\sqrt{\rho_R}^\cT\| (\rho_R^\cT)^{-1} \ox \sigma_E )$ is continuous in $\rho$, the lower semi-continuity of $I_{\alpha}(\cN_{A \rightarrow B}, \rho_A)$ follows directly from Eq.~(\ref{eq:conti}).
\end{IEEEproof}

\medskip
We also need the following technical lemma.
\begin{lemma}
  \label{lem:refunprop}
Let $\cN_{A\rar B}$ be a quantum channel. The function $s\mapsto sI_{1+s}(\cN_{A\rar B})$ is monotonically increasing and convex on $[0,\infty)$.
\end{lemma}

\begin{IEEEproof}
It is shown in~\cite{CMW2016strong} that $s\mapsto I_{1+s}(\cN_{A\rar B})$ is monotonically increasing. This implies that $s\mapsto sI_{1+s}(\cN_{A\rar B})$ is monotonically increasing too. To see that $s\mapsto sI_{1+s}(\cN_{A\rar B})$ is convex, we write $sI_{1+s}(\cN_{A\rar B})=\max\limits_{\rho_A} \{sI_{1+s}(\cN_{A\rar B},\rho_A)\}$ and then note that $s\mapsto sI_{1+s}(\cN_{A\rar B},\rho_A)$ is convex by Ref.~\cite{HayashiTomamichel2016correlation}.
\end{IEEEproof}

\bigskip
\begin{IEEEproof}[Proof of Theorem~\ref{thm:reliability}: Eqs.~(\ref{eq:reliability-u}) and (\ref{eq:reliability-e})]
To prove Eq.~(\ref{eq:reliability-u}), we employ Theorem~\ref{thm:converse}. In Ref.~\cite{LiYao2024reliability}, it is shown that for any state $\rho_{RB}$ and $r\geq0$,
\begin{equation}
\label{eq:expmax}
\lim_{n \rightarrow \infty} \frac{-1}{n} \log \delta_{R^n : B^n}\big(\rho_{RB}^{\ox n}, nr\big)
=\frac{1}{2} \sup_{s \geq 0} \big\{ s\big(r-I_{1+s}(R : B)_{\rho_{RB}}\big) \big\}.
\end{equation}
So Theorem~\ref{thm:converse} and Eq.~(\ref{eq:expmax}) let us obtain that for any bipartite pure state $\varphi_{RA}$,
\begin{align}
       \limsup_{n\rar\infty} \frac{-1}{n} \log P^{\rm{sim}}\big(\cN_{A \rar B}^{\ox n}, nr\big) \nonumber
\leq & \lim_{n\rar\infty} \frac{-1}{n} \log \delta_{R^n : B^n}\big(\cN(\varphi_{RA})^{\ox n}, nr\big) \\
  =  & \frac{1}{2} \sup_{s \geq 0}\big\{s\big(r-I_{1+s}(R:B)_{\cN(\varphi_{RA})}\big)\big\}  \nonumber\\
  =  & \frac{1}{2} \sup_{s \geq 0}\big\{s\big(r-I_{1+s}(\cN_{A\rar B},\varphi_A)\big)\big\}.
      \label{eq:reliability-p1}
\end{align}
With this, our estimation is proceeded as
\begin{align}
     E^{\rm{sim}}(\cN_{A \rar B}, r)
\leq&\limsup_{n\rar\infty}\frac{-1}{n}\log P^{\rm{sim}}\big(\cN_{A\rar B}^{\ox n},nr\big) \nonumber\\
\leq&\frac{1}{2}\min_{\varphi_A\in\cS(A)}\sup_{s \geq 0}\big\{s\big(r-I_{1+s}
                                                   (\cN_{A\rar B},\varphi_A)\big)\big\}   \nonumber\\
  = &\frac{1}{2}\sup_{s \geq 0}\min_{\varphi_A\in\cS(A)}\big\{s\big(r-I_{1+s}
                                                   (\cN_{A\rar B},\varphi_A)\big)\big\}   \nonumber\\
  = &\frac{1}{2}\sup_{s\geq 0}\big\{s\big(r-I_{1+s}(\cN_{A\rar B})\big)\big\}. \label{eq:reliability-p2}
\end{align}
In Eq.~\eqref{eq:reliability-p2}, the first inequality is by the definition of the reliability function $E^{\rm{sim}}$, the second inequality follows from Eq.~(\ref{eq:reliability-p1}) by minimizing its right hand side over all states $\varphi_{A}\in\cS(A)$, and the first equality is by an application of Sion's minimax theorem~\cite{Sion1958general}. To see that Sion's minimax theorem is applicable here, we have (a) the function $s\mapsto s\big(r-I_{1+s}(\cN_{A\rar B},\varphi_A)\big)$ is continuous and concave on the convex set $[0,\infty)$ because $s\mapsto sI_{1+s}(\cN_{A\rar B},\varphi_A)$ is continuous and convex~\cite{HayashiTomamichel2016correlation}, and (b) for any $s\geq0$ the function
$\varphi_A\mapsto s\big(r-I_{1+s}(\cN_{A\rar B},\varphi_A)\big)$ is continuous and convex on the compact convex set $\cS(A)$, by Lemma~\ref{lem:concavity}.

\smallskip
At last, Eq.~(\ref{eq:reliability-e}) holds because when $r\leq R_{\rm{critical}}\equiv\frac{\mathrm{d}}{\mathrm{d}s}sI_{1+s}(\mathcal{N}_{A \rightarrow B})\big|_{s=1}$, the lower bound of Eq.~(\ref{eq:reliability-l}) and the upper bound of Eq.~(\ref{eq:reliability-u}) coincide. To see this, consider the function
\[
f(s)=s\big(r-I_{1+s}(\cN_{A \rar B})\big)
\]
for $s\in[0,\infty)$. The convexity of $s\mapsto sI_{1+s}(\cN_{A \rar B})$ shown in Lemma~\ref{lem:refunprop} implies that $f(s)$ is concave. So, $\sup\limits_{s\geq 0}f(s) = \max\limits_{0\leq s\leq 1}f(s)$ holds if and only if $f'(1)\leq 0$. This condition is equivalent to $r\leq R_{\rm{critical}}$.
\end{IEEEproof}

\section{Conclusion and Discussion}
\label{sec:discussion}
We have derived fundamental achievability and converse bounds for the performance of reverse Shannon simulation of a quantum channel. These bounds translate to the corresponding lower and upper bounds for the reliability function. By showing that the lower bound and upper bound coincide when the classical communication rate $r$ is such that $r\leq R_{\textrm{critical}}$, we have obtained the exact formula of the reliability function for $r$ falling in this range. Remarkably, our result has provided an operational interpretation to the channel's sandwiched R\'enyi information $I_{\alpha}(\cN)$ of order $\alpha\in(1,2]$, in characterizing the direct error exponent. We point out that in~\cite{LiYao2024strong}, an operational interpretation for $I_{\alpha}(\cN)$ with $\alpha\in(1,+\infty)$ has been found by the authors, in characterizing the strong converse exponent of entanglement-assisted communication.

It is well known that with the assistance of free entanglement, two bits of classical communication is equivalent to one qubit of quantum communication, as the consequence of quantum teleportation~\cite{BBCJPW1993teleporting} and dense coding~\cite{BennettWiesner1992communication}. Thus, as indicated in Lemma~\ref{lem:relation}, our solution applies as well to the problem of channel simulation with quantum communication assisted by unlimited free entanglement.

When the classical communication rate is above the critical point, we are unable to obtain the reliability function. In fact, the existence of a critical point, of which at one side the problem remains unsolved, is a common phenomenon in the study of reliability functions~\cite{Gallager1968information,Dalai2013lower, LYH2023tight}. Works in literature (e.g., \cite{Gallager1968information,Shannon1959zero,Lovasz1979on}) shows that combinatorics come into play for channel communication. We do not know whether this is still the case for our simulation setting. We leave it as an open problem. We note that in the independent work~\cite{RTB2023moderate}, the moderate deviation expansion for reverse Shannon simulation of quantum channels in the low-error case has been derived. The second-order asymptotics of this task remains a major open problem. Yet another interesting problem related to our work, is to find the reliability function for entanglement-assisted communication over quantum channels.

\section*{Acknowledgements}
The authors are grateful to one of the anonymous referees for suggesting look at the gap between the critical point and the channel's mutual information, which leads to Example~\ref{exam:crit-mi}. They also would like to thank Yimeng Cao for his assistance in drawing Fig.~\ref{fig:critical-v}.

\bigskip
{\appendix
The following proof of Lemma~\ref{lem:deFa} follows the same idea of~\cite{CKR2009postselection}.

\medskip
\begin{IEEEproof}[Proof of Lemma~\ref{lem:deFa}]
We start with an arbitrary state $\varphi_{TA^n}$ in which $T$ is any reference system. Set
\[
\tilde{\varphi}_{GTA^n}:=\frac{1}{|S_n|}\sum_{\pi\in S_n}\proj{\pi}_G \ox
(\id_T\ox W^\pi_{A^n})(\varphi_{TA^n}).
\]
Then $\tilde{\varphi}_{A^n}$ is invariant under permutations. So it has a symmetric purification
$\tilde{\varphi}_{\bar{R}^nA^n}$, which can be written in the vector form as
\[
\ket{\tilde{\varphi}}_{\bar{R}^nA^n}=\sqrt{(|A|)^n}\big(\1_{\bar{R}^n}\ox \sqrt{\tilde{\varphi}_{A^n}}\big)(\ket{\Phi}_{\bar{R}A})^{\ox n},
\]
with $\bR\cong A$ and $\ket{\Phi}_{\bar{R}A}=\frac{1}{\sqrt{|A|}}\sum_{i=1}^{|A|}\ket{i}_{\bar{R}}\ket{i}_A$ being the maximally entangled state. $\ket{\tilde{\varphi}}_{\bar{R}^nA^n}$ is symmetric in the sense that $W^\pi_{(\bR A)^n}\ket{\tilde{\varphi}}_{\bar{R}^nA^n}=\ket{\tilde{\varphi}}_{\bar{R}^nA^n}$ for any $\pi\in S_n$. We have
\begin{align}
    &1-F\big(\cM(\varphi_{TA^n}), \cE(\varphi_{TA^n})\big) \nonumber \\
=   &1-\sum_{\pi\in S_n} \frac{1}{|S_n|}F\big(\cM(W^\pi_{A^n}(\varphi_{TA^n})),
     \cE(W^\pi_{A^n}(\varphi_{TA^n}))\big) \nonumber \\
=   &1-F\big(\cM(\tilde{\varphi}_{GTA^n}), \cE(\tilde{\varphi}_{GTA^n})\big) \nonumber \\
\leq&1-F\big(\cM(\tilde{\varphi}_{\bar{R}^nA^n}), \cE(\tilde{\varphi}_{\bar{R}^nA^n})\big), \label{eq:deF-1}
\end{align}
where the second line is due to the symmetry of $\cM_{A^n\rar B^n}$ and $\cE_{A^n\rar B^n}$ as well as the fact that the fidelity function is invariant under the action of unitary operations, the third line is by direct calculation, and the last line follows from the monotonicity of the fidelity function under CPTP maps.

Let
\[
\zeta_{\bR^nA^n}:=\int\phi_{\bR A}^{\ox n}\di\nu(\phi_{\bR A}),
\]
where the integration is over all pure states on $\cH_{\bR A}$ and $\nu(\cdot)$ is the unique uniform spherical measure satisfying $\nu(\phi)=\nu(u\phi u^*)$ for any unitary transformation $u$ on $\cH_{\bR A}$. Then by Schur's lemma, $\zeta_{\bR^nA^n}$ is the maximally mixed state supported on the symmetric subspace
\[\mathrm{Sym}(\cH_{\bar{R}A}^{\ox n})
:=\left\{\ket{v}\in\cH_{\bR A}^{\ox n}:W^\pi_{(\bR A)^n}\ket{v}=\ket{v}, \forall\ \pi\in S_n\right\}.
\]
Let $\{\ket{\tilde{\varphi}^{(a)}}_{\bar{R}^nA^n}\}_{a=1}^{g_{n,|A|}}$ be an orthonormal basis of $\mathrm{Sym}(\cH_{\bar{R}A}^{\ox n})$, with the identification $\ket{\tilde{\varphi}^{(1)}}_{\bar{R}^nA^n}\equiv\ket{\tilde{\varphi}}_{\bar{R}^nA^n}$. Note that the dimension of $\mathrm{Sym}(\cH_{\bar{R}A}^{\ox n})$ is exactly $g_{n,|A|}$. So, we have $\zeta_{\bR^nA^n}=\frac{1}{g_{n,|A|}}\sum_{a=1}^{g_{n,|A|}}\tilde{\varphi}^{(a)}_{\bar{R}^nA^n}$ and its purification can be written as
\[
\ket{\zeta}_{\tR\bR^nA^n}=\frac{1}{\sqrt{g_{n,|A|}}}\sum_{a=1}^{g_{n,|A|}}\ket{a}_{\tR}
\ket{\tilde{\varphi}^{(a)}}_{\bar{R}^nA^n},
\]
where $\{\ket{a}_{\tR}\}$ is an orthonormal basis of $\tR$. We now have
\begin{align}
     &F\big(\cM(\zeta_{\tR\bR^nA^n}), \cE(\zeta_{\tR\bR^nA^n})\big) \nonumber\\
\leq &F\Big(\cM\big(\frac{1}{{g_{n,|A|}}}\sum_a\proj{a}_{\tR}\ox
            \tilde{\varphi}^{(a)}_{\bar{R}^nA^n}\big),
            \cE\big(\frac{1}{{g_{n,|A|}}}\sum_a\proj{a}_{\tR}\ox \tilde{\varphi}^{(a)}_{\bar{R}^nA^n}\big)\Big)      \nonumber \\
  =  &\frac{1}{{g_{n,|A|}}}F\big(\cM(\tilde{\varphi}_{\bar{R}^nA^n}),
                                 \cE(\tilde{\varphi}_{\bar{R}^nA^n})\big)
      +\frac{1}{{g_{n,|A|}}}\sum_{a=2}^{g_{n,|A|}}
                           F\big(\cM(\tilde{\varphi}_{\bar{R}^nA^n}^{(a)}),
                                 \cE(\tilde{\varphi}_{\bar{R}^nA^n}^{(a)})\big) \nonumber \\
\leq &\frac{1}{{g_{n,|A|}}}F\big(\cM(\tilde{\varphi}_{\bar{R}^nA^n}),
                                 \cE(\tilde{\varphi}_{\bar{R}^nA^n})\big)
      +\frac{g_{n,|A|}-1}{g_{n,|A|}}, \label{eq:deF-2}
\end{align}
where the first inequality comes from the monotonicity of the fidelity function under CPTP maps.

From Eq.~(\ref{eq:deF-1}) and Eq.~(\ref{eq:deF-2}) we get that
\begin{align*}
     &g_{n,|A|}\Big(1-F\big(\cM(\zeta_{\tR\bR^nA^n}), \cE(\zeta_{\tR\bR^nA^n})\big)\Big) \\
\geq &1-F\big(\cM(\tilde{\varphi}_{\bar{R}^nA^n}), \cE(\tilde{\varphi}_{\bar{R}^nA^n})\big) \\
\geq &1-F\big(\cM(\varphi_{TA^n}), \cE(\varphi_{TA^n})\big).
\end{align*}
Since the state $\varphi_{TA^n}$ was chosen arbitrarily, the above relation lets us obtain
\begin{align*}
     &P(\cM_{A^n\rar B^n},\cE_{A^n\rar B^n}) \\
  =  &\sqrt{1-\min_{\varphi_{TA^n}} F\big(\cM(\varphi_{TA^n}), \cE(\varphi_{TA^n})\big)^2} \\
\leq &\sqrt{\Big(1-\min_{\varphi_{TA^n}} F\big(\cM(\varphi_{TA^n}), \cE(\varphi_{TA^n})\big)\Big)
            \Big(1+F\big(\cM(\zeta_{\tR\bR^nA^n}),\cE(\zeta_{\tR\bR^nA^n})\big)\Big)} \\
\leq &\sqrt{g_{n,|A|}\Big(1-F\big(\cM(\zeta_{\tR\bR^nA^n}),\cE(\zeta_{\tR\bR^nA^n})\big)\Big)
            \Big(1+F\big(\cM(\zeta_{\tR\bR^nA^n}),\cE(\zeta_{\tR\bR^nA^n})\big)\Big)} \\
  =  &\sqrt{g_{n,|A|}}\ P\big(\cM(\zeta_{\tR\bR^nA^n}),\cE(\zeta_{\tR\bR^nA^n})\big).
\end{align*}
Noticing that the de Finetti state $\zeta_{R_nA^n}$ in Lemma~\ref{lem:deFa} is identical to $\zeta_{\tR\bR^nA^n}$ by setting $R_n\equiv\tR\bR^n$, we complete the proof.
\end{IEEEproof}

\bigskip
The following proposition is about properties of the lower and upper bounds of the reliability function.
\begin{proposition}
  \label{prop:propertyreli}
For the quantum channel $\cN_{A\rar B}$, we have
\begin{enumerate}[1.]
\item when $r> R_{\rm{critical}}$, the lower bound $E_l(r)=\frac{1}{2}\max\limits_{0\leq s\leq 1}
  \{s(r-I_{1+s}(\cN))\}$ becomes linear: $E_l(r)=\frac{1}{2}(r-I_2(\cN))$;
\item when $r>I_\text{max}(\cN)\equiv\lim\limits_{s\rar\infty}I_{1+s}(\cN)$, the
  upper bound $E_u(r)=\frac{1}{2} \sup\limits_{s \geq 0}\{s(r-I_{1+s}(\cN))\}$ is infinity.
\end{enumerate}
\end{proposition}

\begin{IEEEproof}
(1) Since $s\mapsto sI_{1+s}(\cN)$ is convex on $[0,\infty)$ by Lemma~\ref{lem:refunprop}, the function $f(s)=s(r-I_{1+s}(\cN))$ is concave. When $f'(1)>0$, which is equivalent to $r> R_{\rm{critical}}$, we have $E_l(r)=\frac{1}{2} f(1)=\frac{1}{2}(r-I_2(\cN))$. (2) The function $s\mapsto I_{1+s}(\cN)$ is monotonically increasing~\cite{CMW2016strong} and obviously bounded for $s\in[0,\infty)$. So, the limit $\lim\limits_{s\rar\infty}I_{1+s}(\cN)$ exists and the statement follows.
\end{IEEEproof}
}

\bigskip
\begin{figure}[ht]
  \includegraphics[width=8cm]{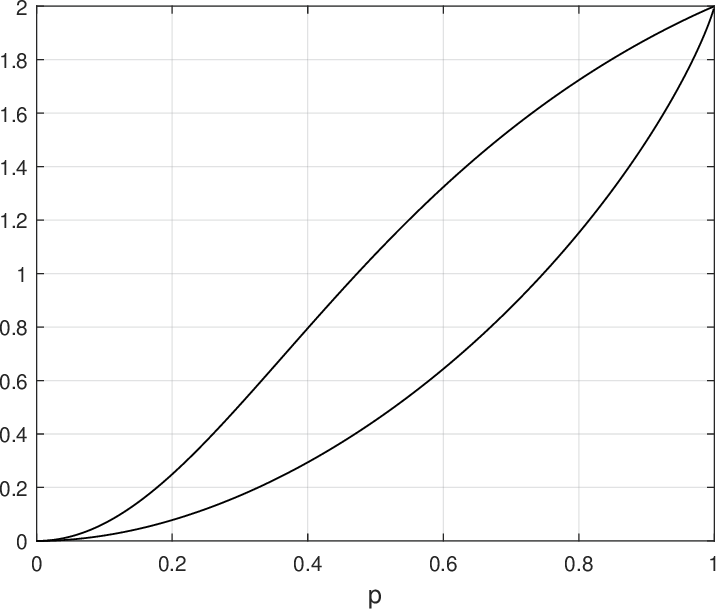} \centering
  \caption{Critical value and mutual information of the qubit depolarizing channel. Upper curve is $R_{\rm{critical}}(\cN^{(p)})$, and lower curve is $I(\cN^{(p)})$.
  }
  \label{fig:critical-v}
\end{figure}

In the following example, we calculate the mutually information and the critical value for the depolarizing channels.
\begin{example}
\label{exam:crit-mi}
We consider the depolarizing channels $\cN_{A\rar B}^{(p)}$ with $|A|=|B|=d$ and $p\in[0,1]$, given by
\[
\cN^{(p)}(\rho):=(1-p)\rho + p\frac{\1}{d}.
\]
Due to the symmetry of this channel, namely, its being covariant with respect to the Heisenberg-Weyl group, we can easily verify that for $s\geq 0$,
\[
I_{1+s}(\cN^{(p)})
=D_{1+s}\left(p\Phi_{AB}+(1-p)\frac{\1_A}{d}\ox\frac{\1_B}{d}\Big\|\frac{\1_A}{d}\ox\frac{\1_B}{d}\right),
\]
where $\Phi_{AB}$ is a maximally entangled pure state. With this, we get by direct computation that
\begin{align*}
I(\cN^{(p)})&=\frac{d^2p+1-p}{d^2}\log(d^2p+1-p)+\frac{(d^2-1)(1-p)}{d^2}\log(1-p), \\
R_{\rm{critical}}(\cN^{(p)})&=\frac{(d^2p+1-p)^2\log(d^2p+1-p)}{(d^2p+1-p)^2+(d^2-1)(1-p)^2}
                             +\frac{(d^2-1)(1-p)^2\log(1-p)}{(d^2p+1-p)^2+(d^2-1)(1-p)^2}.
\end{align*}

In Fig.~\ref{fig:critical-v}, we depict these two quantities as functions of $p$, in the special qubit case $d=2$. It shows that $R_{\rm{critical}}(\cN^{(p)})$ is considerably larger than $I(\cN^{(p)})$. On the other hand, when $d$ is big, $R_{\rm{critical}}(\cN^{(p)})$ can be much larger than $I(\cN^{(p)})$, since for constant $p$ we have
\begin{align*}
I(\cN^{(p)}) &\sim 2p\log d, \\
R_{\rm{critical}}(\cN^{(p)}) &\sim 2\log d.
\end{align*}
\end{example}


\end{document}